\documentclass[a4paper]{article}
\usepackage[english]{babel}
\usepackage[T1]{fontenc}
\usepackage[utf8]{inputenc}
\usepackage{amssymb,amsmath,amsthm}
\usepackage{microtype}
\usepackage{color}
\usepackage{graphicx}
\usepackage{authblk}

\theoremstyle{plain}
\newtheorem{theorem}{Theorem}
\newtheorem{lemma}[theorem]{Lemma}

\newtheorem{corollary}[theorem]{Corollary}
\newtheorem{conjecture}[theorem]{Conjecture}

\title{Conflict-Free Coloring of Intersection Graphs\footnote{This work was partially supported by the DFG Research Unit "Controlling Concurrent Change", funding number FOR 1800, project FE407/17-2, "Conflict Resolution and Optimization".}}

\author[1]{Sándor P. Fekete}
\author[1]{Phillip Keldenich}
\affil[1]{Algorithms Group, TU Braunschweig, Germany\\
	\texttt{s.fekete,p.keldenich@tu-bs.de}
}
\date{}

\begin{document}
\maketitle
\begin{abstract}
A {\em conflict-free $k$-coloring} of a graph $G=(V,E)$ assigns one of $k$ different colors to some of the vertices such that, 
for every vertex $v$, there is a color that is assigned to exactly one vertex among $v$ and $v$'s neighbors. 
Such colorings have applications in wireless networking, robotics, and geometry, and are well studied in graph theory.
Here we study the conflict-free coloring of geometric intersection graphs.  
We demonstrate that the intersection graph of $n$ geometric objects without fatness properties and
size restrictions may have conflict-free
chromatic number in $\Omega(\log n/\log\log n)$ and in $\Omega(\sqrt{\log n})$ for disks or squares of different sizes; 
it is known for general graphs that the worst case is in $\Theta(\log^2 n)$. 
For unit-disk intersection graphs, we prove that it is NP-complete
to decide the existence of a conflict-free coloring
with one color; we also show that six colors always suffice,
using an algorithm that colors unit disk graphs of restricted height with two colors.  
We conjecture that four colors are sufficient, which we prove for unit squares instead of unit disks.
For interval graphs, we establish a tight worst-case bound of two.
\end{abstract}

\section{Introduction}
Coloring the vertices of a graph is one of the fundamental problems in graph theory, both scientifically and historically.
The notion of proper graph coloring can be generalized to hypergraphs in several ways.
One natural generalization is \emph{conflict-free coloring}, which asks to color the vertices of a hypergraph such that every hyperedge has at least one uniquely colored vertex.
This problem has applications in wireless communication, where ``colors'' correspond to different frequencies.

The notion of conflict-free coloring can be brought back to simple graphs, e.g., by considering as hyperedges the neighborhoods of the vertices of $G$.
The resulting problem arises in certain variants of frequency assignment problems if one is not interested in achieving signal coverage for all points in a region, but only at certain points of interest.
For an illustration, consider a scenario in which one has a given set of nodes in the plane and wants to establish a communication network between them.
Moreover, assume that constructing nodes at new locations is either very expensive or forbidden, and one can only ``upgrade'' any existing node to a wireless base station.

Conflict-free coloring also plays a role in robot navigation, where different beacons are used for providing direction.
To this end, it is vital that in any given location, a robot is adjacent to a beacon with a frequency that is unique among the ones that can be received. 

Both in the frequency assignment setting and in the robot navigation setting, one typically wants to avoid placing unnecessary base stations or beacons.
Abstractly speaking, this corresponds to leaving some vertices uncolored, yielding the following formalization of \emph{conflict-free coloring} of graphs.
For any vertex $v\in V$ of a simple graph $G=(V,E)$, the neighborhood $N[v]$ consists of all vertices adjacent to $v$ and $v$ itself.
A \emph{conflict-free $k$-coloring} of $G$ assigns one of $k$ colors to a (possibly proper) subset $S \subseteq V$ of vertices such that every vertex $v \in V$ has a uniquely colored neighbor.
The \emph{conflict-free chromatic number} $\chi_{CF}(G)$ of $G$ is the smallest $k$ for which a conflict-free coloring exists.
Depending on the situation it may also be more natural to consider \emph{open neighborhood conflict-free coloring}, where each vertex $v$ must have a uniquely colored neighbor in its \emph{open neighborhood} $N(v)$ not including $v$.

Conflict-free coloring has received an increasing amount of attention.
Because of the motivation arising from frequency assignment, it is natural to investigate the conflict-free coloring of intersection
graphs, in particular, of simple shapes such as disks or squares.
In addition, previous work has considered either general graphs and hypergraphs (e.g., see \cite{pt-cfcgh-09})
or other geometric scenarios in the presence of obstacles (e.g., see \cite{hoffmann_et_al:LIPIcs:2015:5097}); we give a more detailed overview further down.
This adds to the relevance of conflict-free coloring of intersection graphs, which lie in the intersection of general graphs and geometry.

There is a spectrum of different scientific challenges when studying conflict-free coloring.
What are worst-case bounds on the necessary number of colors?
When is it NP-hard to determine the existence of a conflict-free $k$-coloring?
We address these questions for the case of intersection graphs.

\medskip
{\bf Our contribution.} We present the following results.
\begin{itemize}
\item We demonstrate that $n$ geometric objects without fatness properties and size restrictions may induce intersection graphs with conflict-free chromatic number
in $\Omega(\log n/\log\log n)$.

\item We prove that non-unit square and disk graphs may require $\Omega(\sqrt{\log n})$ colors. 
Deciding conflict-free $k$-colorability is NP-hard for any $k$ for these graph classes.

\item It is NP-complete for unit-disk intersection graphs to decide the existence of a conflict-free coloring with one color.
The same holds for intersection graphs of unit squares and other shapes.

\item Six colors are always sufficient for conflict-free coloring of unit disks.
This uses an algorithm that colors unit disk graphs contained in a strip of restricted height with two colors.

\item Using a similar argument, we prove that four colors are always sufficient for conflict-free coloring of unit squares. 

\item As a corollary, we get a tight worst-case bound of two on the conflict-free chromatic number of interval graphs.
\end{itemize}

{\bf Related work.}
In the geometric context, motivated by frequency assignment problems, the study of conflict-free coloring of hypergraphs was initiated by Even~et~al.~\cite{elrs-cfcsg-03} and Smorodinsky~\cite{s-cpcg-03}.
For disk intersection hypergraphs, Even et al.~\cite{elrs-cfcsg-03} prove that $\mathcal{O}(\log n)$ colors suffice.
For disk intersection hypergraphs with degree at most $k$, Alon and Smorodinsky~\cite{as-cfcsd-06} show that $\mathcal{O}(\log^3 k)$ colors are sufficient.
Cheilaris et al.~\cite{css-piclcfcgh-11} prove that $\mathcal{O}(\log n)$ colors suffice for the case of hypergraphs induced by planar Jordan regions with linear union complexity.
If every edge of a disk intersection hypergraph must have $k$ distinct unique colors, Horev et al.~\cite{hks-cfcms-10} prove that $\mathcal{O}(k\log n)$ suffice.
Moreover, for unit disks, Lev-Tov and Peleg~\cite{lp-cfcud-09} present an $\mathcal{O}(1)$-approximation algorithm for the conflict-free chromatic number.
Abam et al.~\cite{abp-ftcfc-08} consider the problem of making a conflict-free coloring robust against removal of a certain number of vertices, and prove worst-case bounds for the number of colors required.
The online version of the problem has been studied as well; see, e.g., the work of Chen et al.~\cite{cfk+-ocfci-06} that presents, among other results, a randomized online algorithm using $\mathcal{O}(\log n)$ colors to maintain a conflict-free coloring of a set of intervals, or the work of Bar-Nov et al.~\cite{bco+-ocfch-10} that presents an online algorithm using $\mathcal{O}(k\log n)$ colors for some classes of $k$-degenerate hypergraphs.

The dual problem in which one has to color a given set of points such that each region contains a uniquely colored point has also received some attention.
Har-Peled and Smorodinsky~\cite{hs-cfcpsrp-05} prove that for families of pseudo-disks, every set of points can be colored using $\mathcal{O}(\log n)$ colors.
For rectangles, Ajwani et al.~\cite{aegr-cfcrro-07} show that $\mathcal{O}(n^{0.382})$ colors suffice, whereas Elbassioni and Mustafa~\cite{em-cfcrr-06} show that it is possible to add a sublinear number of points such that sublinearly many colors suffice.
For coloring points on a line with respect to intervals, Cheilaris et al.~\cite{cgrs-scfci-14} present a 2-approximation algorithm, and a $\left(5-\frac{2}{k}\right)$-approximation algorithm when every interval must contain $k$ uniquely colored points.

In general hypergraphs, i.e., hypergraphs not necessarily arising from a geometric context, Ashok et al.~\cite{adk-efmcfch-15} prove that maximizing the number of conflict-freely colored hyperedges is FPT with respect to the number of conflict-freely colored hyperedges in the solution.

Conflict-free coloring also arises in the context of the conflict-free variant of the chromatic Art Gallery Problem, which asks to guard a polygon using colored guards such that each point sees a uniquely colored guard.
Fekete et al. \cite{ffh-cgcagp-2014} prove that computing the chromatic number is NP-hard in this context.
On the positive side, Hoffman et al.~\cite{hoffmann_et_al:LIPIcs:2015:5097} prove $\Theta(\log\log n)$ colors are sometimes necessary and always sufficient for the conflict-free chromatic art gallery problem under rectangular visibility in orthogonal polygons.
For straight-line visibility, Bärtschi et al.~\cite{bs-cfcagc-14} prove that $\mathcal{O}(\log n)$ colors suffice for orthogonal and monotone polygons and $\mathcal{O}(\log^2 n)$ colors suffice for simple polygons.
This is generalized by Bärtschi et al.~\cite{bgm+-ibcfcagp-2014}, who prove that $\mathcal{O}(\log n)$ colors suffice for simple polygons.

There also has been work regarding the scenario where the hypergraph is induced by the neighborhoods of vertices of a simple graph.
Except for the need to color all vertices, this corresponds to the scenario considered in this work.
This does not change the asymptotic number of colors required, since it suffices to insert one additional color to color all vertices that would otherwise remain uncolored.
In this situation, Pach and Tardos~\cite{pt-cfcgh-09} prove that the conflict-free chromatic number of an $n$-vertex graph is in $\mathcal{O}(\log^2 n)$.
Glebov et al.~\cite{gst-cfcg-14} extend this result by proving that almost all $G(n,\omega(1/n))$-graphs have conflict-free chromatic number $\mathcal{O}(\log n)$.
Moreover, they show that the upper bound of Pach and Tardos~\cite{pt-cfcgh-09} is tight by giving a randomized construction for graphs having conflict-free chromatic number $\Theta(\log^2 n)$. 
In more recent work, Gargano and Rescigno~\cite{gr-ccfcg-15} show that finding the conflict-free chromatic number for general graphs is NP-complete, and prove that the problem is FPT w.r.t.~vertex cover or neighborhood~diversity~number.
In our work with Abel et al.~\cite{aad+-tcscfcpg-17}, we consider conflict-free coloring of general and planar graphs and proved a conflict-free variant of Hadwiger's conjecture, which implies that planar graphs have conflict-free chromatic number at most three.
Most recently, Keller and Smorodinsky~\cite{ks-cfciggo-17} consider conflict-free coloring on intersection graphs of geometric objects, in a scenario very similar to ours.
Among other results, they prove that $\mathcal{O}(\log n)$ colors suffices to color a family $\mathcal{F}$ of pseudodisks in a conflict-free manner.
With respect to \emph{open} neighborhoods (also known as \emph{pointed neighborhoods}), they prove that this is tight; for closed neighborhoods as studied in this paper, the tightness of this bound is not proven and remains open.
They also consider the list coloring variant of the problem.

Conflict-free coloring has also been studied for other graph-based hypergraphs.
For instance, Cheilaris and Tóth~\cite{ct-gumcfc-11} consider the case of hypergraphs induced by the paths of a graph.
If the input is the graph, they prove that it is coNP-complete to decide whether a given coloring is conflict-free.

Conflict-free coloring is not the only type of coloring for which unit disk graphs have been found to require a bounded number of colors.
In their recent work, McDiarmid~et~al.~\cite{dmp-ccgg-17} consider clique
coloring of unit disk graphs, in particular with regard to the asymptotic behavior
of the clique chromatic number of random unit disk graphs.  They also prove
that every unit disk graph in the plane can be colored with nine colors, while
three colors are sometimes necessary.
Similar to the present paper, they prove this by cutting the plane into strips of height $\sqrt{3}$; for each of these strips it is then proven that three colors
suffice.

\section{Preliminaries}
In the following, $G = (V,E)$ denotes a graph on $n := |V|$ vertices.
For a vertex $v$, $N(v)$ denotes its \emph{open neighborhood} and $N[v] = N(v) \cup \{v\}$ denotes its \emph{closed neighborhood}.
A \emph{conflict-free $k$-coloring} of a graph $G = (V,E)$ is a coloring $\chi: V' \to \{1,\ldots,k\}$ of a subset $V' \subseteq V$ of the vertices of $G$, such that each vertex $v$ has at least one \emph{conflict-free neighbor} $u \in N[v]$, i.e., a neighbor $u$ whose color $\chi(u)$ occurs only once in $N[v]$.
The \emph{conflict-free chromatic number} $\chi_{CF}(G)$ is the minimum number of colors required for a conflict-free coloring of $G$.

A graph $G$ is a \emph{disk graph} iff $G$ is the intersection graph of disks in the plane.
$G$ is a \emph{unit disk graph} iff $G$ is the intersection graph of disks with fixed radius $r=1$ in the plane,
and a \emph{unit square graph} iff $G$ is the intersection graph of axis-aligned squares with side length 1 in the plane.
A unit disk (square) graph is of \emph{height} $h$ iff $G$ can be modeled by the intersection of unit disks (squares) with center points in $(-\infty,\infty) \times [0,h]$.
In the following, when dealing with intersection graphs, we assume that we are given a geometric model.
In the case of unit disk and unit square graphs, we identify the vertices of the graph with the center points of the corresponding geometric objects in this model.

\section{General Objects}
\label{sec:general-objects}
Intersection graphs of geometric objects can generally contain cliques of arbitrary size, so their chromatic number may be unbounded.
However, cliques do not require a large number of colors in a conflict-free coloring, so it is not immediately clear whether the intersection graphs for a family of geometric objects have bounded conflict-free chromatic number.

If the intersecting objects can be scaled down arbitrarily, i.e., if every representable graph can be represented using arbitrarily small area, we can make use of the following lemma to prove lower bounds on the number of colors required.
\begin{lemma}
	\label{lem:gk_gadgets}
	Let $G_k$ be a graph with $\chi_{CF}(G_k) \geq k$, and let $G$ be a graph containing two disjoint copies $J_k^1$ and $J_k^2$ of $G_k$.
	Let $v_1,\ldots,v_l$ be vertices of $G$, not contained in $J_k^1$ or $J_k^2$, and let each vertex $v_i$ be adjacent to every vertex of $J_k^1$ and $J_k^2$.
	Moreover, let these vertices be the only neighbors of $J_k^1$ and $J_k^2$.
	Then in every conflict-free $k$-coloring of $G$, one of the vertices $v_1,\ldots,v_l$ has a color that appears only once in $v_1,\ldots,v_l$.
\end{lemma}
\begin{proof}
	Assume there was a conflict-free $k$-coloring $\chi$ of $G$ such that none of the vertices $v_1, \ldots, v_l$ has a unique color.
	Therefore, each vertex in $J_k^1$ has a conflict-free neighbor in $J_k^1$, and restricting $\chi$ to $V(J_k^1)$ yields a conflict-free $k$-coloring of $J_k^1$.
	As $\chi_{CF}(G_k) \geq k$, each color is used on $V(J_k^1)$ at least once.
	The same holds for $J_k^2$.
	Therefore, each vertex $v_1, \ldots, v_l$ has at least two occurrences of each color in its neighborhood; this contradicts the fact that $\chi$ is a conflict-free coloring of $G$.
\end{proof}

For general objects like freely scalable ellipses or rectangles, it is possible
to model a complete graph $K_n$ of arbitrary size $n$, such that the following
conditions hold: (1) For every object $v$, there is some non-empty area of $v$
not intersecting any other objects.
(2) For every pair of objects $v,w$, there is a non-empty area common to these objects not intersecting any other objects.
This can be seen by choosing $n$ intersecting lines such that no three lines intersect in a common point.
These lines can then be approximated using sufficiently thin objects to achieve the desired configuration.

In this case, the conflict-free chromatic number is unbounded, because we can
inductively build a family $G_n$ of intersection graphs with $\chi_C(G_n) = n$
as follows.
Starting with $G_1 = (\{v\},\emptyset)$ and $G_2 = C_4$ (a four-vertex cycle), we
construct $G_n$ by starting with a $K_n$ modeled according to conditions (1) and (2).
For every object $v$, we place two scaled-down non-intersecting copies of $G_{n-1}$ into an area covered only by $v$; Figure~\ref{fig:unbounded-slim} depicts the construction of $G_5$ for ellipses.
According to Lemma~\ref{lem:gk_gadgets}, these gadgets enforce that every vertex of the underlying $K_n$ is colored.
For every pair of objects $v,w$, we place two scaled-down non-intersecting copies of $G_{n-2}$ into an area covered only by $v$ and $w$.
Using an argument similar to that used in the proof of Lemma~\ref{lem:gk_gadgets}, these gadgets enforce that $v$ and $w$ have to receive different colors.
Thus the resulting graph requires $n$ colors.
\begin{figure}
	\begin{center}
		\resizebox{.55\linewidth}{!}{\includegraphics{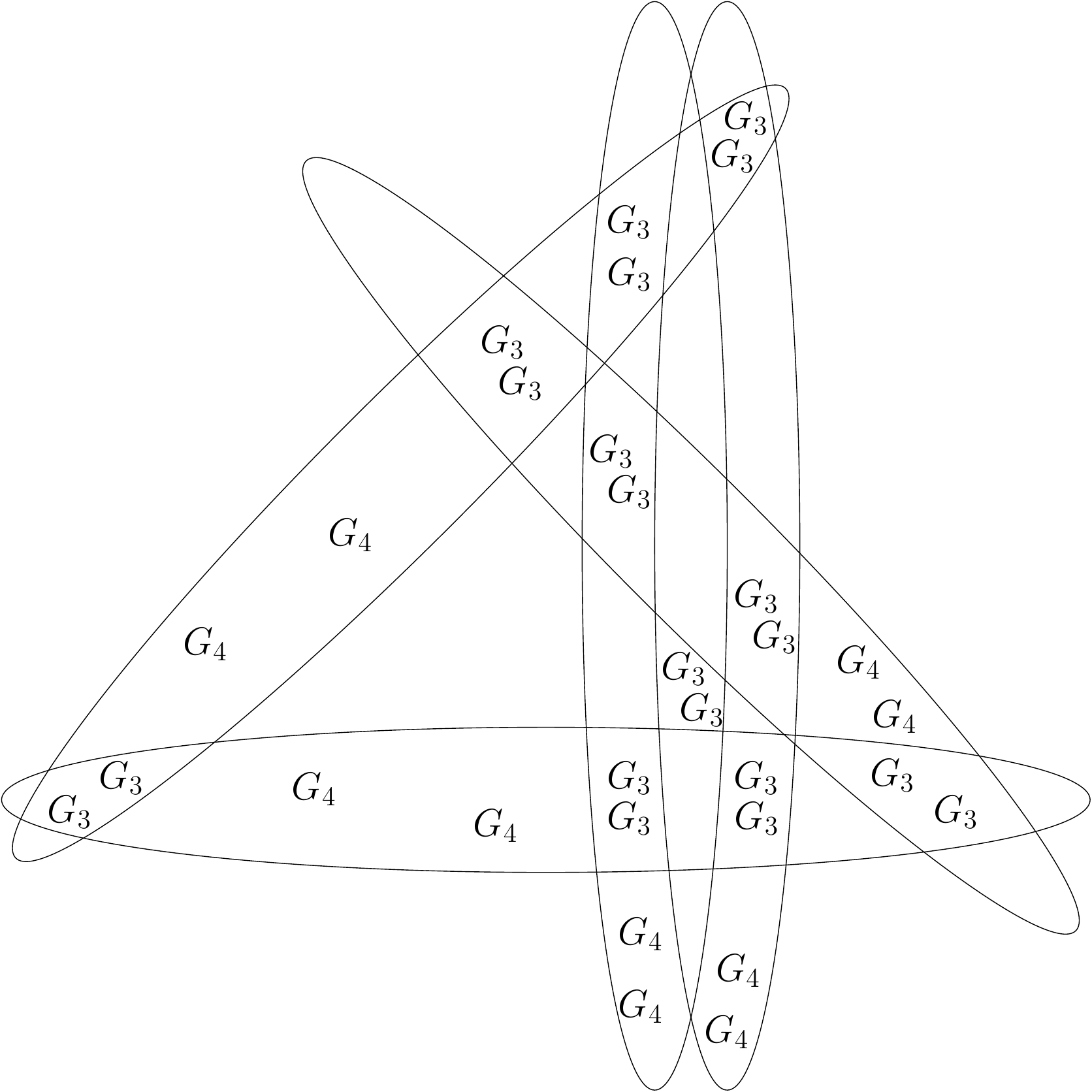}}
	\end{center}
	\caption{The graph $G_5$, shown as an intersection graph of ellipses, requires 5 colors.}
	\label{fig:unbounded-slim}
\end{figure}

The number of vertices used by this construction satisfies the recurrence $$|G_1| = 1,\quad |G_2| = 4,\quad |G_n| = n + 2n|G_{n-1}| + n(n-1)|G_{n-2}|.$$
To estimate the growth of $|G_n|$, let $\bar{G}_n = |G_n|$ for $n \leq 2$ and $\bar{G}_n = 3n\bar{G}_{n-1} + n(n-1)\bar{G}_{n-2}$; clearly, $\bar{G}_n \geq |G_n|$ for all $n$.
The recurrence $\bar{G}_n$ has the closed-form solution
\begin{eqnarray*}
	\bar{G}_n & = & \frac{n!}{13 \cdot 2^{n+1}} \cdot \Big((5\sqrt{13}-13)(3+\sqrt{13})^n - (13+5\sqrt{13})(3-\sqrt{13})^n\Big)\\
		  & = & \mathcal{O}\left(n!\left(\frac{3+\sqrt{13}}{2}\right)^n\right),
\end{eqnarray*}

implying that the number of colors required in geometric intersection graphs on $n$ vertices may be $\Omega(\frac{\log n}{\log \log n}).$

We summarize.

\begin{theorem}
The intersection graph of $n$ convex objects in the plane may have conflict-free chromatic number in $\Omega(\log n/\log\log n)$.
\label{thm:logloglog}
\end{theorem}

The best upper bound on the number of colors required in this scenario that we are aware of is $\mathcal{O}(\log^2 n)$, which holds for general graphs and is due to Pach and Tardos~\cite{pt-cfcgh-09}.

\section{Different-Sized Squares and Disks}
	Due to their fatness, squares and disks do not allow us to construct an arbitrarily big clique $K_n$ such that condition (2) of Section~\ref{sec:general-objects} holds.
	However, we can still prove that there is no constant bound on their conflict-free chromatic number.
	The proof is based on Lemma~\ref{lem:gk_gadgets}, which enables us to reduce the conflict-free coloring problem on intersection hypergraphs to our problem.
	\begin{theorem}
		The conflict-free chromatic number of disk intersection graphs and square intersection graphs can be $\Omega(\sqrt{\log n})$.
		\label{thm:disk-graphs-unbounded}
	\end{theorem}
	\begin{proof}
		We begin our proof by inductively constructing, for any number of colors $k$, a disk intersection graph $D_k$ with conflict-free chromatic number $\chi_{CF}(D_k) = k$ and $\mathcal{O}(2^{2k^2})$ vertices.
		The first level of the construction is $D_1$, consisting of an isolated vertex.
		The remainder of the construction is based on a lower-bound example due to Even et al.~\cite{elrs-cfcsg-03}, requiring $\Omega(\log n)$ colors when each point in the union of all disks must lie in a uniquely colored disk.
		This lower-bound example consists of \emph{chain disks} $1, \ldots, 2^{k-1}$ on a horizontal line segment, placed such that all disks overlap in the center.
		For each interval $[i,j]$, $1 \leq i \leq j \leq 2^{k-1}$, there is one region with non-zero area in which exactly the disks from this interval overlap.
		This situation is depicted in Figure~\ref{fig:gk_chains}.
		\begin{figure}
			\begin{center}
				\resizebox{.7\linewidth}{!}{\includegraphics{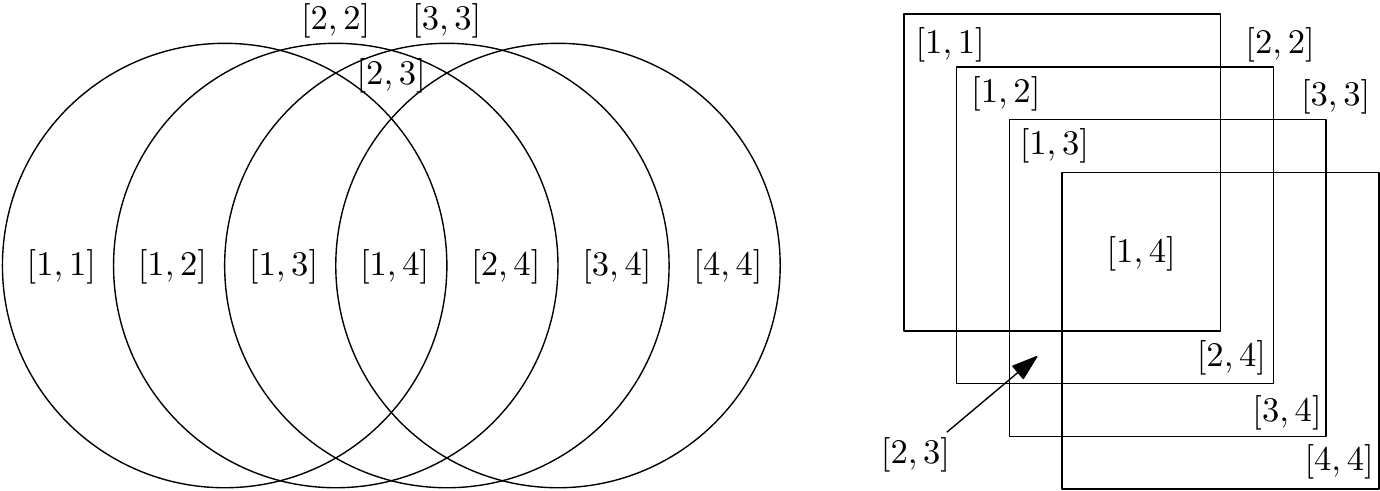}}
			\end{center}
			\caption{A chain of length 4 using either disks or squares, requiring 3 colors in every conflict-free hypergraph coloring. Adding two copies of $D_2$ to every interval yields a disk intersection graph that requires 3 colors in any conflict-free coloring.}
			\label{fig:gk_chains}
		\end{figure}
		
		To construct $D_k$, for each such interval $[i,j]$, we choose one such region and place two scaled-down disjoint copies of $D_{k-1}$ in it.
		We prove that $D_k$ requires $k$ colors by induction on $k$.
		That $D_1$ requires one color is clear.
		Given that $D_{k-1}$ requires $k-1$ colors for some $k \geq 2$, we can prove that $D_k$ requires $k$ colors as follows.
		Assume there was a conflict-free $(k-1)$-coloring $\chi$ of $D_k$.
		Due to Lemma~\ref{lem:gk_gadgets}, we know that, for every interval $[i,j]$, $1 \leq i \leq j \leq 2^{k-1}$, at least one of the chain disks in $[i,j]$ has a unique color.
		We now prove using induction that any color assignment with this property has at least $k$ colors.
		For a chain of length $2^0$, one color is required for the interval $[1,1]$.
		For a chain of length $2^l$, we require one unique color for the interval $[1,2^l]$.
		Let $i$ be the chain disk colored using this color.
		At least one of the intervals $[1,i-1],[i+1,2^l]$ has length at least $2^{l-1}$.
		By induction, this interval requires $l$ colors.
		These colors must all be distinct from the color used for $i$, therefore forcing us to use $l+1$ colors in total.
		This contradicts the fact that $\chi$ uses only $k-1$ colors; therefore, $\chi_{CF}(D_k) \geq k$.
		The number of vertices used by $D_k$ satisfies the recurrence $$|D_1| = 1,\quad |D_k| = 2^{k-1} + \frac{2^{k-1}(2^{k-1}+1)}{2}|D_{k-1}| = 2^{k-1} + (2^{2k-3} + 2^{k-2})|D_{k-1}|,$$
		which is in $\mathcal{O}(2^{2k^2})$.
		All our arguments can also be applied to squares instead of disks.
	\end{proof}

	\begin{theorem}
		For any fixed number of colors $k$, deciding whether a disk (or square) intersection graph is conflict-free $k$-colorable is NP-complete.
	\end{theorem}
	\begin{proof}
		Conflict-free coloring of disk intersection graphs is clearly in NP.
		We prove NP-hardness inductively by reducing conflict-free $(k-1)$-colorability to conflict-free $k$-colorability.
		For $k = 1$, NP-hardness follows from Theorem~\ref{thm:npc1}.
		In order to reduce $(k-1)$-colorability to $k$-colorability, consider a graph $G$ for which conflict-free $(k-1)$-colorability is to be decided. We construct a graph $H$ that is conflict-free $k$-colorable iff $G$ is conflict-free $(k-1)$-colorable.
		For the sake of simplicity, we will refer to the disk representation of $G$ and $H$ in the description of our construction; however, the construction does not require this representation unless a disk representation of $H$ is desired.
		We start construction of $H$ with a chain $H_1$ of length $2^k$.
		For every interval $[i,j], 1 \leq i \leq j \leq 2^k$, the hypergraph $H_1$ contains a corresponding hyperedge.
		By the argument used in the proof of Theorem~\ref{thm:disk-graphs-unbounded}, in a conflict-free hypergraph coloring of $H_1$, $k+1$ colors are required.
		For each $1 \leq i \leq 2^k$, we choose a region corresponding to the interval $[i,i]$ and place two disjoint copies of $G$ inside.
		For all other intervals $[i,j], i \neq j$, we choose a region corresponding to the interval and place one copy of $G$ inside; see Figure~\ref{fig:hardness-k-colors}.
		\begin{figure}
			\begin{center}
				\resizebox{.7\linewidth}{!}{\includegraphics{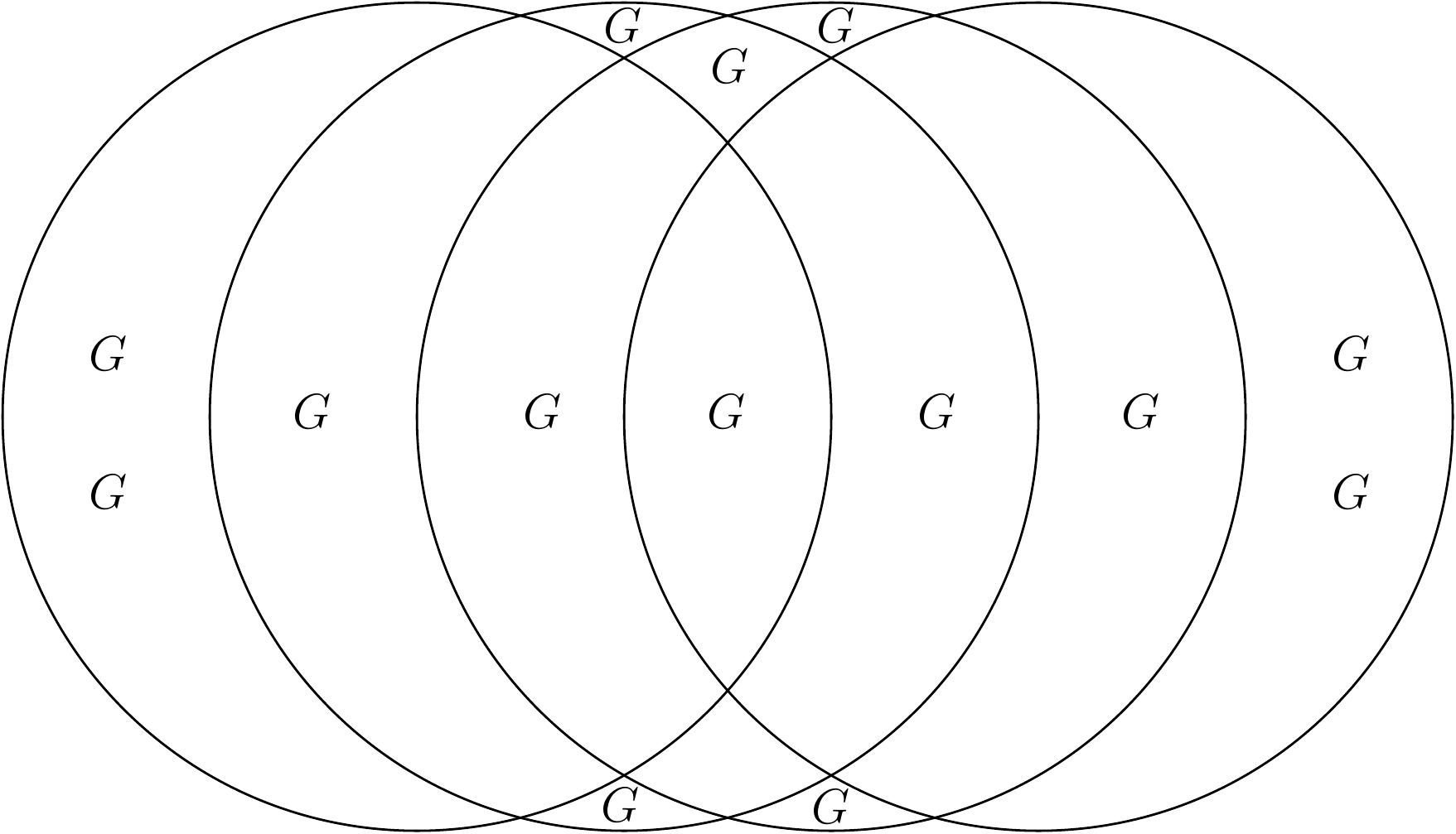}}
			\end{center}
			\caption{
				The setup of the reduction for $k=2$ colors. If $G$ is conflict-free $1$-colorable, every copy of $G$ can be colored using just color $2$.
				Every chain disk can then receive a conflict-free neighbor for instance by coloring one of the chain disks with color $1$.
			}
			\label{fig:hardness-k-colors}
		\end{figure}

		If $G$ is conflict-free $(k-1)$-colorable, $H$ is conflict-free $k$-colorable as follows.
		Each copy of $G$ can be colored using colors $\{2,\ldots,k\}$ in a conflict-free manner.
		In order to give each chain vertex a conflict-free neighbor, in one of the two copies of $G$ in interval $[i,i]$, we can color an arbitrary vertex of $G$ with color $1$.

		If $G$ is not conflict-free $(k-1)$-colorable, by Lemma~\ref{lem:gk_gadgets} we have to color all chain vertices.
		Given such a coloring of the chain vertices, there are two possibilities for each interval $[i,j], 1 \leq i < j \leq 2^k$.
		(1)~The chain vertices corresponding to the interval are colored such that there is a uniquely colored vertex.
		(2)~Every color occurring on the corresponding chain vertices occurs more than once.

		As $H_1$ does not have a conflict-free hypergraph $k$-coloring, case (2) must occur for at least one interval $[i,j]$.
		Let $G'$ be the copy of $G$ placed in $[i,j]$.
		Vertices in $G'$ cannot use the chain vertices as conflict-free neighbors.
		Moreover, at least one color already appears twice in the neighborhood of every vertex of $G'$ and thus cannot be the color of a conflict-free neighbor of any vertex in $G'$.
		Therefore a conflict-free $k$-coloring of $H$ would yield a conflict-free $(k-1)$-coloring of $G$, which is a contradiction.

		For fixed $k$, this construction works in polynomial time.
		Therefore, conflict-free $k$-coloring of disk intersection graphs is NP-complete for any fixed $k \geq 1$.
		All our arguments can also be applied to squares instead of disks.
	\end{proof}
	
	In \cite{elrs-cfcsg-03}, Even et al. prove that $\Theta(\log n)$ colors are always sufficient and sometimes necessary to color a disk intersection hypergraph in a conflict-free manner.
	This implies that $\mathcal{O}(\log n)$ colors are sufficient in our case, leaving a gap of $\mathcal{O}(\sqrt{\log n})$.

\section{Unit-Disk Graphs}
The construction used in the previous section hinges on high aspect ratios of the intersecting shapes.
In the setting of frequency assignment for radio transmitters, it is natural to only consider {\em fat} objects with bounded aspect ratio, such as unit disks and unit squares.
As it turns out, their intersection graphs have conflict-free chromatic number bounded by a small constant; on the other hand, even deciding the existence of a conflict-free coloring with a single color is NP-complete.

\subsection{Complexity: One Color}
\label{subsec:one}
While it is trivial to decide whether a graph has a regular chromatic number of 1 and straightforward
to check a chromatic number of 2, it is already NP-complete to decide whether a conflict-free coloring
with a single color exists, even for unit-disk intersection graphs. This is a refinement
of Theorem~4.1 in Abel et al.~\cite{aad+-tcscfcpg-17}, which shows the same results for general planar graphs.

\begin{theorem}
It is NP-complete to decide whether a unit-disk intersection graph $G=(V,E)$ has a conflict-free coloring with one color.
\label{thm:npc1}
\end{theorem}
\begin{proof}
Membership in NP is straightforward.
We prove NP-hardness by reduction from {\sc Positive Planar 1-in-3-SAT}, see Mulzer and Rote~\cite{mr-mwtnph-08}.
Given a Boolean formula $\phi$ in 3-CNF with only positive literals and planar clause-variable incidence graph, we construct a unit disk intersection graph $G(\phi)$ that has a conflict-free $1$-coloring iff $\phi$ is 1-in-3-satisfiable.
Let $\phi$ consist of variables $\{x_1,\ldots,x_n\}$ and clauses $\{c_1,\ldots,c_k\}$.
In $G(\phi)$, variables $x_i$ are represented by a cycle of length $12k$ and clauses are represented by a clause gadget; see Figure~\ref{fig:npc1-clause-and-variable}.
\begin{figure}
	\begin{center}
		\resizebox{.85\linewidth}{!}{\includegraphics{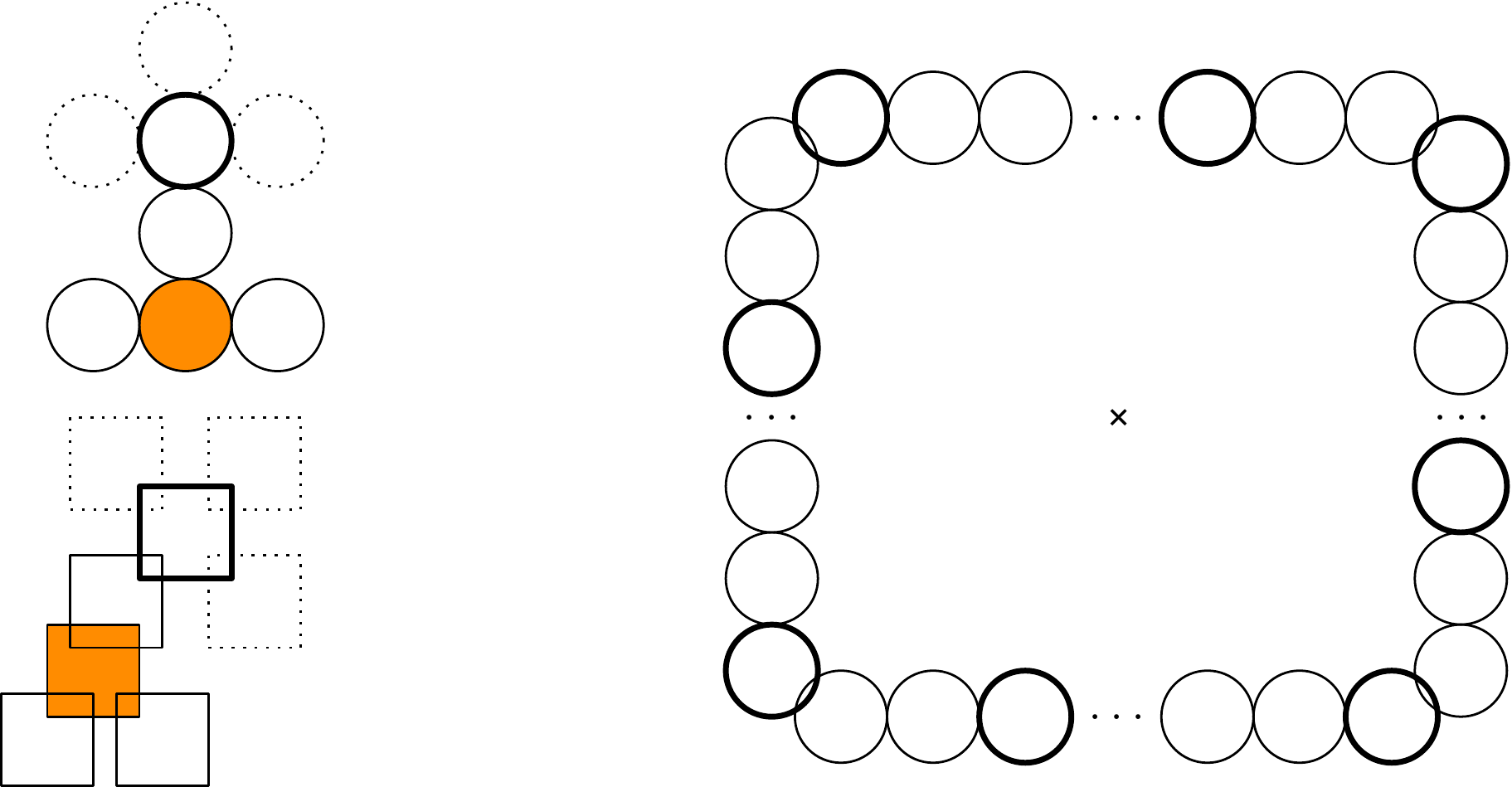}}
	\end{center}
	\caption{
		{\bf (Left)}
			Clause gadget represented using unit disks and unit squares.
			The \emph{clause vertices} that are attached to the remainder of $G(\phi)$ are drawn with bold outline.
			Dashed objects depict where the connections to the variables attach to the clause vertex.
			Orange vertices must be colored in any conflict-free $1$-coloring; therefore, the clause vertex must remain uncolored.
		{\bf (Right)}
		A variable gadget. \emph{True vertices}, i.e., vertices that are colored if the variable is set to {\em true} are drawn with bold outline.
		In a conflict-free coloring of a variable gadget, every third vertex along the cycle must be colored.
		This implies that we must color either all {\em true} vertices or none of them.
	}
	\label{fig:npc1-clause-and-variable}
\end{figure}
The clause gadgets have distinguished \emph{clause vertices}.
Every third vertex of a variable gadget is a \emph{true vertex}; these are the vertices that are colored if the corresponding variable is set to {\em true}.
We connect the clause vertex of each clause to a {\em true} vertex of each variable occurring in the clause using paths of length divisible by 3.
There are sufficiently many {\em true} vertices in each variable gadget to avoid having to use a {\em true} vertex for more than one clause.

A clause vertex $c$ cannot be colored and cannot have a conflict-free neighbor in the clause gadget.
Its neighbors $d_1,d_2,d_3$ along the paths connecting $c$ to variables require conflict-free neighbors themselves.
Therefore either $d_i$ must be colored itself or its predecessor $p_i$ on the path must be colored.
Along the path to the variable, every third vertex must be colored, starting with either $d_i$ or $p_i$.
Thus it is only possible to color $d_i$ if the {\em true} vertices of the corresponding variable are colored; $p_i$ can only be colored if the {\em true} vertices of the corresponding variable are not colored, see Figure~\ref{fig:variable-clause-connection}.

\begin{figure}
	\begin{center}
		\resizebox{.4\linewidth}{!}{\includegraphics{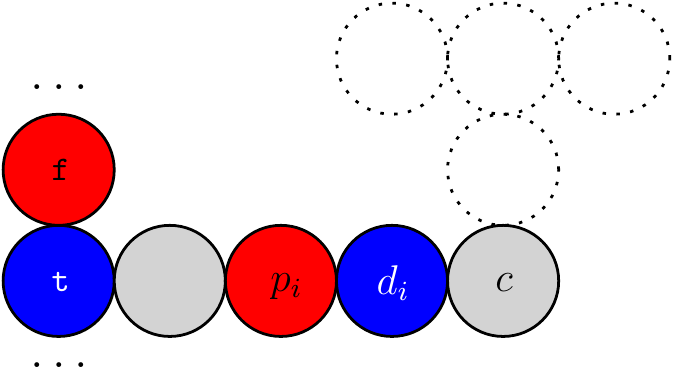}}
	\end{center}
	\caption{
		Path connecting a variable to a clause vertex $c$ (with dashed clause gadget).
		The vertices marked {\tt t} and {\tt f} are {\em true} and {\em false} vertices of a variable. Blue vertices are colored if the variable is set to {\em true}, red vertices are colored if the variable is set to {\em false}.
		Gray vertices cannot be colored.
		For $c$, this is enforced directly by the clause gadget; for the other vertices along the path, it follows from the fact that the clause vertex cannot be colored.
	}
	\label{fig:variable-clause-connection}
\end{figure}

In a conflict-free $1$-coloring $\chi$ of $G(\phi)$, each clause vertex has a conflict-free neighbor.
This implies that for each clause $c_j$, exactly one of the vertices $d_1,d_2,d_3$ is colored, which in turn implies that exactly one variable gadget connected to $c_j$ has all its {\em true} vertices colored.
Therefore setting the variables that have all their {\em true} vertices colored to {\em true} yields a 1-in-3-satisfying assignment for $\phi$.
Analogously, a 1-in-3-satisfying assignment for $\phi$ induces a conflict-free $1$-coloring of $G(\phi)$.

We still have to argue that $G(\phi)$ is a unit disk graph that can be constructed in polynomial time.
We consider a straight-line embedding of $\phi$'s clause-variable incidence graph with vertices placed on an $\mathcal{O}(n+k) \times \mathcal{O}(n+k)$ integer grid; such an embedding can be obtained in polynomial time based on the graph-drawing techniques of Fraysseix, Pach and Pollack~\cite{fpp-hdpgg-90}.
We enlarge this embedding by an appropriate polynomial factor to ensure that clauses and variables are far enough from each other and edges do not come too close to gadgets they are not incident to.
We use the embedding of the variables and clauses to place the center of variables gadgets (marked by a cross in Figure~\ref{fig:npc1-clause-and-variable}) and clause vertices accordingly.
The edges between variables and clauses can then be replaced by paths; we can ensure that their length is divisible by 3 by simple local modifications (shifting a constant number of disks closer to each other and inserting a constant number of new disks).
When deciding which {\em true} vertex to use for a clause and where to place the disk adjacent to the clause vertex, we preserve the order of edges around the clause and variables vertices of the embedding.
If multiple possible {\em true} vertices exist, we use a closest one.
This excludes intersections between different paths close to variables and unintended intersections between paths and variable gadgets.
We may have to bend the paths around the clause gadgets; however, for every clauses this involves only a constant number of disks.
\end{proof}

\subsection{A Worst-Case Upper Bound: Six Colors}
\label{subsec:six}

On the positive side, we show that the conflict-free chromatic number of unit disk graphs is bounded by six.
We do not believe this result to be tight.
In particular, we conjecture that the number is bounded by four; in fact, we do not even know an example for which two colors are insufficient.
One of the major obstacles towards obtaining tighter bounds is the fact that a simple graph-theoretic characterization of unit disk graphs is not available, as recognizing unit disk graphs is complete for the existential theory of the reals \cite{km-sdprg-12}.
This makes it hard to find unit disk graphs with high conflict-free chromatic number, especially considering the size that such a graph would require:
The smallest graph with conflict-free chromatic number three we know has 30 vertices, and by enumerating all graphs on 12 vertices one can show that at least 13 vertices are necessary, even without the restriction to unit disk graphs.

One approach to conflict-free coloring of unit disk graphs is by subdividing the plane into strips, coloring each strip independently.
We conjecture the following.

\begin{conjecture}
\label{con:s2-cfc2}
Unit disk graphs of height $2$ are conflict-free $2$-colorable.
\end{conjecture}

If this conjecture holds, every unit disk graph is conflict-free 4-colorable.
In this case, one can subdivide the plane into strips of height 2, and then color the subgraphs in all even strips using colors $\{1,2\}$ and the subgraphs in odd strips using colors $\{3,4\}$.
Instead of Conjecture~\ref{con:s2-cfc2}, we prove the following weaker result.

\begin{theorem}
Unit disk graphs $G$ of height $\sqrt{3}$ are conflict-free 2-colorable.
\label{thm:strips-height-s3}
\end{theorem}
\begin{proof}
Given a realization of $G$ consisting of unit disks with center points with $y$-coordinate in $[0,\sqrt{3}]$, we compute a conflict-free 2-coloring of $G$ using the following simple greedy approach.
In an order corresponding to the lexicographic order of the points in $\mathbb{R}^2$ (denoted by $\leq$), we build a set $C$ of colored vertices to which we alternatingly assign colors 1 and 2.
In each step, we add a new point to $C$ until all points are covered, i.e., they are either colored or have a colored neighbor.
In order to select the next colored point, we find the lexicographically maximal point $c$ such that every point $c' < c$ is already colored or has a colored neighbor in $C \cup \{c\}$.
We observe that this point $c$ may already have a colored neighbor, but then there must be an uncovered point between $c$ and previously colored point.

In this procedure, every point $v$ is assigned a colored neighbor $w \in N[v]$.
It remains to exclude the following three cases.
(1) A colored point $v$ is adjacent to another point $w$ of the same color,
(2) an uncolored point is adjacent to two or more points of one color and none of the other color,
(3) an uncolored point is adjacent to two or more points of both colors.

To this end, we use the following observation.
Each colored point $c$ induces a closed vertical \emph{strip} of width 2 centered around $c$.
As shown in Figure~\ref{fig:strips_height_s3_substrips}, every point $v$ in this strip is adjacent to $c$.
\begin{figure}
	\begin{center}
		\resizebox{.95\linewidth}{!}{\includegraphics{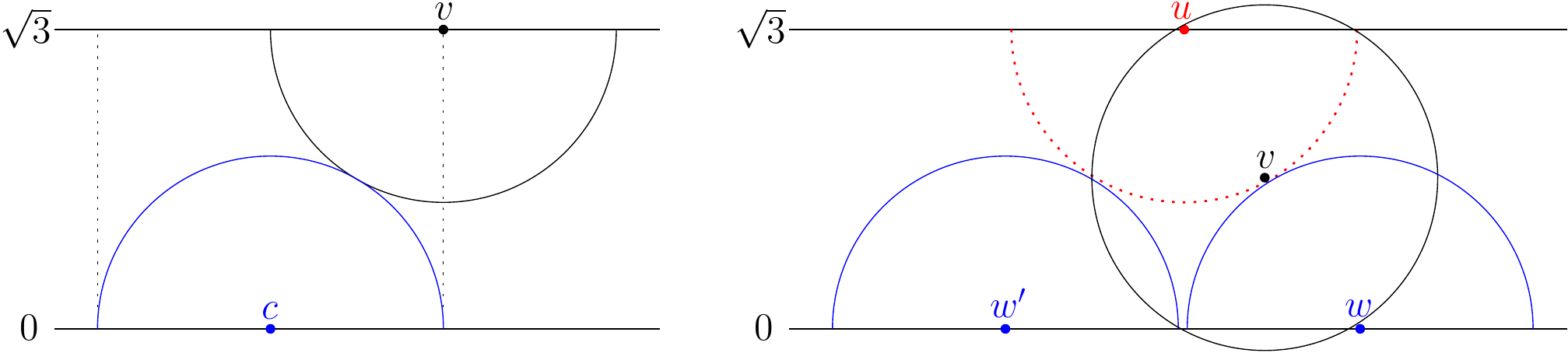}}
	\end{center}
	\caption{{\bf (Left)} Every colored point $c$ induces a vertical strip of width 2 (dashed lines); all points $v$ within this strip are adjacent to $c$. {\bf (Right)} The configuration in case (2); there must be a point $u$ of color 2 adjacent to $v$.}
	\label{fig:strips_height_s3_substrips}
\end{figure}
Thus, the horizontal distance between two colored points must be greater than 1.
For case (1), assume there was a point $v$ of color 1 adjacent to a point $w > v$ of color 1.
This cannot occur, because between $v$ and $w$, there must be a point $x$ of
color 2; therefore, the horizontal distance between $v$ and $w$ must be greater
than 2, a contradiction.

Regarding case (2), assume there was an uncolored point $v$ adjacent to two points $w' < v < w$ of color 1; see Figure~\ref{fig:strips_height_s3_substrips}.
Between points $w'$ and $w$, there must be a point $u$ of color 2, and $v$ must not be adjacent to $u$.
There are two possible orderings: $w' < v < u < w$ and $w' < u < v < w$.
W.l.o.g., let $u < v$; the other case is symmetric.
In this situation, the $x$-coordinates of the points have to satisfy
$x(u) < x(v)-1$, $x(w') < x(u)-1$, and thus $x(w') < x(v)-2$ in contradiction to the assumption that $v$ and $w'$ are adjacent.

\begin{figure}
	\begin{center}
		\resizebox{.7\linewidth}{!}{\includegraphics{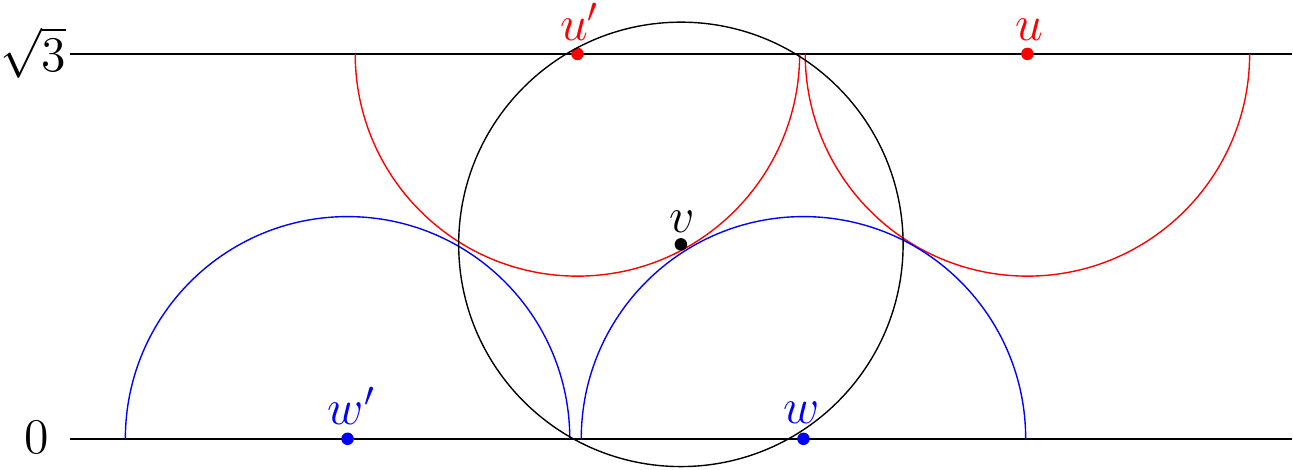}}
	\end{center}
	\caption{The configuration in case (3); the algorithm would have chosen $v$ instead of $u'$.}
	\label{fig:strips_s3_case3}
\end{figure}

Regarding case (3), assume there was an uncolored point $v$ adjacent to two points $w' < v < w$ of color 1 and two points $u' < v < u$ of color 2.
W.l.o.g., assume $w' < u' < v < w < u$ as depicted in Figure~\ref{fig:strips_s3_case3}; the case $u' < w'$ is symmetric.
Because $w'$ and $v$ are adjacent, the vertical strip induced by $v$ intersects the strip induced by $w'$.
Thus, there cannot be a point $y$ with $w' < y < v$ not adjacent to $w'$ or $v$.
This is a contradiction to the choice of $u'$: The algorithm would have chosen $v$, or a larger point, instead of $u'$.
\end{proof}
The following corollary follows by subdividing the plane into strips of height $\sqrt{3}$.

\begin{corollary}
Unit disk graphs are conflict-free $6$-colorable.
\end{corollary}

Unfortunately, the proof of Theorem~\ref{thm:strips-height-s3} does not appear to have a straightforward generalization to strips of larger height.
Further reducing the height to find strips that are colorable with one color is also impossible, see Section~\ref{sec:interval-graphs}.

\subsection{Unit-Disk Graphs of Bounded Area}
Proving Conjecture~\ref{con:s2-cfc2} is non-trivial, even when all center points lie in a $2 \times 2$-square.
In this setting, a circle packing argument can be used to establish the sufficiency of three colors.
If a unit disk graph with conflict-free chromatic number 3 can be embedded into a $2 \times 2$-square, the following are necessary.
	(1) Every minimum dominating set $D$ has size 3, and every pair of dominating vertices must have a common neighbor not shared with the third dominating vertex.
		Thus, every minimum dominating set lies on a 6-cycle without chords connecting a vertex with the opposite vertex.
	(2) $G$ has diameter 2; otherwise, one could assign the same color to two vertices at distance 3.

Using the domination number, one can further restrict the position of the points in the $2 \times 2$-square:
There is an area in the center of the square, depicted in Figure~\ref{fig:conjecture_squares}, that cannot contain the center of any disk because this would yield a dominating set of size 2.
\begin{figure}
	\hfill
	\begin{minipage}{.31\linewidth}
		\resizebox{\textwidth}{!}{\includegraphics{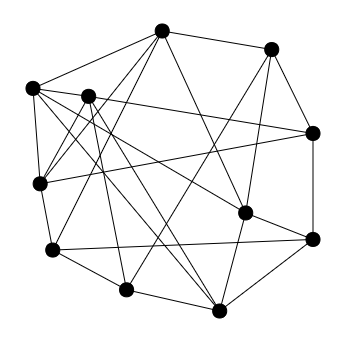}}
	\end{minipage}
	\hfill
	\begin{minipage}{.38\linewidth}
		\resizebox{\textwidth}{!}{\includegraphics{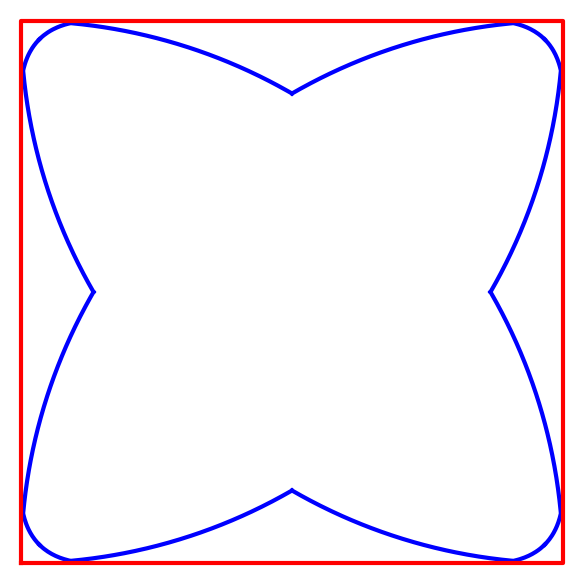}}
	\end{minipage}
	\hfill
	\ \\
	\caption{{\bf (Left)} A vertex-minimal graph satisfying (1) and (2).
	{\bf (Right)} In any unit disk graph $G$ embeddable in a $2 \times 2$-square with domination number $\gamma(G)=3$, no points lie in the depicted area.}
	\label{fig:conjecture_squares}
\end{figure}

The smallest graph satisfying constraints (1) and (2) has 11 vertices and is depicted in Figure~\ref{fig:conjecture_squares}.
It is not a unit disk graph and it is still conflict-free 2-colorable, but every coloring requires at least four colored vertices, proving that coloring a minimum dominating set can be insufficient.
This implies that a simple algorithm like the one used in the proof of Theorem~\ref{thm:strips-height-s3} will most likely be insufficient for strips of greater height.
We are not aware of any unit disk graph satisfying these constraints.

\section{Unit-Square and Interval Graphs}
The constructions of the previous section can also be applied to the case of squares; for interval graphs, we get a tight worst-case bound.

\subsection{Complexity: One Color}
It is straightforward to see that the construction of Theorem~\ref{thm:npc1} can be applied for unit square instead of unit disks.

\begin{corollary}
It is NP-complete to decide whether a unit square graph $G=(V,E)$ has a conflict-free coloring with one color.
\label{thm:npc_sq1}
\end{corollary}

\subsection{A Worst-Case Upper Bound: Four Colors}
The proof of Theorem~\ref{thm:strips-height-s3} can be applied to unit square graphs of height 2 instead of unit disk graphs of height $\sqrt{3}$; see Figure~\ref{fig:squaregraph_strip}.

\begin{corollary}
	Unit square graphs of height $2$ are conflict-free $2$-colorable.
	Unit square graphs are conflict-free $4$-colorable.
\end{corollary}
\begin{figure}[h]
	\begin{center}
		\resizebox{.58\linewidth}{!}{\includegraphics{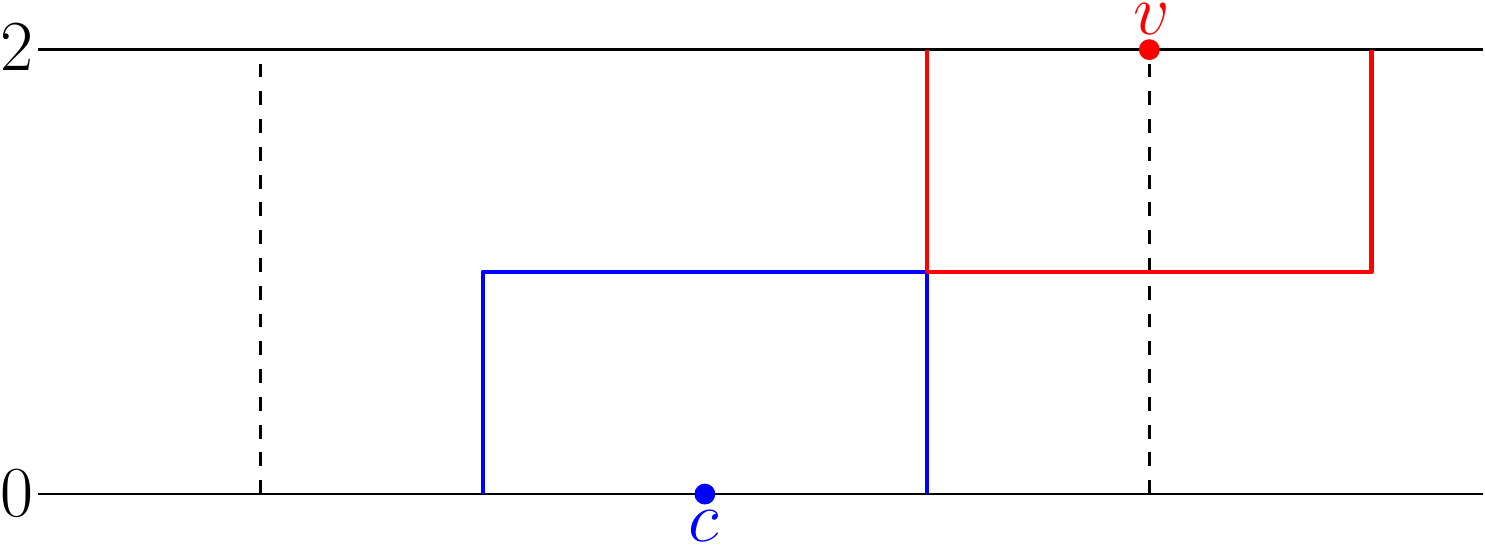}}
	\end{center}
\vspace*{-3mm}
	\caption{
		For unit square graphs of height 2, we have a similar situation to that depicted in Figure~\ref{fig:strips_height_s3_substrips}:
		Centered around each colored vertex $c$, there is a vertical strip of width 4 such that all vertices $v$ with center points in this strip are adjacent to $c$.
		The remainder of the proof of Theorem~\ref{thm:strips-height-s3} applies to unit square graphs analogously.
	}
	\label{fig:squaregraph_strip}
\vspace*{-3mm}
\end{figure}

\subsection{Interval Graphs: Two Colors}
\label{sec:interval-graphs}
Unit interval graphs correspond to unit disk or unit square graphs with all centers lying on a line. 
Even then, two colors in a conflict-free coloring may be required; the Bull Graph is such an example, see Figure~\ref{fig:bull}.

\begin{figure}[h]
	\begin{center}
		\resizebox{0.98\linewidth}{!}{\includegraphics{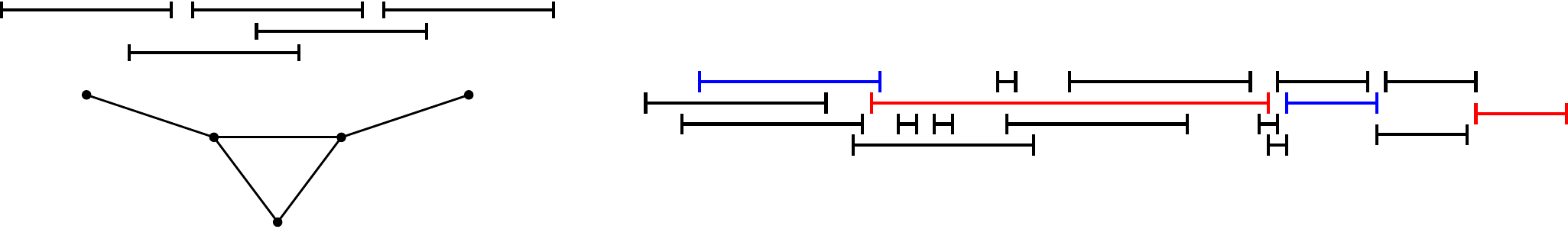}}
	\end{center}
\vspace*{-3mm}
\caption{{\bf (Left)} Realizing the Bull Graph as a unit interval graph. Conflict-free coloring requiring two colors: there is no dominating vertex, the only pair of vertices at distance 3 is no dominating set.
	{\bf (Right)} A conflict-free 2-coloring of an interval graph as computed by the greedy coloring algorithm sketched above.}
	\label{fig:bull}
\vspace*{-3mm}
\end{figure}

In this case, the bound of 2 is tight:
By Theorem~\ref{thm:strips-height-s3}, unit interval graphs are conflict-free $2$-colorable.
By adapting the algorithm used in the proof to always choose the interval extending as far as possible to the right without leaving a previous interval uncovered, this can be extended to interval graphs with non-unit intervals.
For an example of this procedure, refer to Figure~\ref{fig:bull}.

\section{Conclusion}
There are various directions for future work.
In addition to closing the worst-case gap for unit disks (and proving Conjecture~\ref{con:s2-cfc2}), the worst-case conflict-free chromatic number for unit square graphs also remains open.
Other questions include a tight bound for disk (or square) intersection graphs, and a necessary criterion for a family of geometric objects to have intersection graphs with unbounded conflict-free chromatic number.

\bibliography{refs}
\end{document}